\documentclass[conference, onecolumn]{IEEEtran}
\usepackage{graphicx} 
\usepackage[utf8]{inputenc}
\usepackage[a4paper, margin=0.5in]{geometry}

\usepackage{cite, tikz}
\usetikzlibrary{positioning, arrows.meta}
\usepackage{tikz-cd}
\tikzstyle{box}=[rectangle, draw, text centered]
\tikzstyle{line} = [draw, -latex']
\usepackage{pgfplots}
\usepackage{amsmath,amssymb,amsfonts}
\usepackage{algorithmic}
\usepackage{textcomp}
\usepackage{xcolor}
\usepackage{subcaption}

\def\BibTeX{{\rm B\kern-.05em{\sc i\kern-.025em b}\kern-.08em
    T\kern-.1667em\lower.7ex\hbox{E}\kern-.125emX}}
\usepackage{amsmath,amssymb,amsthm}
\usepackage{tikz,enumerate}
\usepackage{cite,mathtools}
\theoremstyle{plain}

\newtheorem{theorem}{Theorem} 
\newtheorem{proposition}{Proposition}

\newtheorem{lemma}{Lemma}

\theoremstyle{remark}

\newtheorem{remark}{Remark}

\theoremstyle{definition}

\allowdisplaybreaks

\title{An Improved Upper Bound on the Capacity of the Primitive Diamond Relay Channel}
\date{\today}
\IEEEoverridecommandlockouts

\author{
\IEEEauthorblockN{Deniz G\"und\"uz}
\IEEEauthorblockA{
\textit{Department of Electrical and}\\
\textit{Electronic Engineering}\\
\textit{Imperial College London}\\
London, UK\\
\texttt{d.gunduz@imperial.ac.uk}
}
\and
\IEEEauthorblockN{Yanxiao Liu}
\IEEEauthorblockA{
\textit{Department of Electrical and}\\
\textit{Electronic Engineering}\\
\textit{Imperial College London}\\
London, UK\\
\texttt{y.liu2@imperial.ac.uk}
}
\and
\IEEEauthorblockN{Yi Liu}
\IEEEauthorblockA{
\raisebox{-0.5\baselineskip}[0pt][0pt]{%
\textit{Department of Information Engineering}}\\
\\
\textit{The Chinese University of Hong Kong}\\
Hong Kong, China\\
\texttt{ly023@ie.cuhk.edu.hk}
}
\thanks{The authors are listed in alphabetical order.}
}
\begin{document}
\maketitle

\begin{abstract}
    We study the capacity of the primitive diamond relay channel and derive an improved upper bound compared with existing upper bounds by employing the Csisz\'ar--K\"orner--Marton sum identity and Gallager-type identification.
    For Gaussian primitive diamond relay channels, we show that our proposed bound coincides with the existing strengthened Gaussian upper bound of Wu, \"Ozg\"ur, Peleg, and Shamai.
    For discrete memoryless channels, our bound can be strictly tighter than both the cut-set bound and the existing upper bound, and we assess the improvement using different channel instances, together with a decode--forward/compress--forward time-sharing inner bound.  
\end{abstract}

\section{Introduction}

Relay networks provide a fundamental model for communication systems in which geographically separated terminals cooperate to convey information from a source to a destination. The diamond relay network was studied in~\cite{schein2001distributed}. A relatively simple yet nontrivial variant is the primitive diamond relay channel. 
In this model, a source transmits over a broadcast channel to two relays, each of which communicates with the destination through an orthogonal, error-free link of finite capacity. This model captures the tension between distributed reception at the relays and limited-rate communication toward a common destination. It also serves as a useful abstraction for distributed antenna systems, cloud radio access networks, and communication networks with capacity-limited fronthaul~\cite{SanderovichShamaiSteinbergKramer2008,SanderovichShamaiSteinberg2009,SanderovichSomekhPoorShamai2009,park2014fronthaul,wu2015cloud,aguerri2019capacity,liu2026nonasymptotic}. Previous work has focused on the complementary diamond setting, in which the source-to-relay links are orthogonal finite-capacity links and the relay-to-destination channel is a multiple-access channel, and has established inner and outer bounds for this model; see, e.g., \cite{kang2015gaussian,bidokhti2016capacity,GK26}. In contrast, the study of the capacity of the present setting was initiated in~\cite{wu2019new}, with subsequent achievability results developed in~\cite{aoe25,gdy25,gdy26}.

\subsection{Channel Model}
Consider the primitive diamond relay communication system depicted in Fig.~\ref{fig:1}. A transmitter wishes to communicate reliably with a receiver through a two-stage network. The first stage is a broadcast channel $T_{Y_1,Y_2|X}$ from the transmitter to the two relay nodes. The second stage consists of two orthogonal, noiseless links of capacities $C_1$ and $C_2$ from the respective relays to the receiver.

We assume that the broadcast channel is memoryless and of the product-form: $$T(y_1,y_2|x) = T(y_1|x)T(y_2|x).$$

For a blocklength-$n$ code, relay $i$ observes $Y_i^n$ and produces an index
\[
    J_{i,n}=\phi_{i,n}(Y_i^n)
    \in\left[1:2^{nC_i}\right],
    \qquad i\in\{1,2\}.
\]
The receiver estimates the source message from $(J_{1,n},J_{2,n})$. A detailed description of this model can also be found in~\cite{wu2019new}.

\begin{figure}[h!]
    \centering
    \begin{tikzpicture}[scale=1, every node/.style={scale=0.74}]
\node at (0.4,-0.6) {$M$}; 
\draw [->,thick] (0,-0.8) -- (0.8,-0.8);
\draw (0.8, -0.4) rectangle +(1.6,-0.8); \node at (1.6,-0.75) {Transmitter};
\draw [->,thick] (2.4,-0.8)-- (3.4,-0.8); \node at (2.9, -0.6) {$ X^n$};
\draw (3.4, -0.4) rectangle +(1.6,-0.8); \node at (4.2,-0.8) {$T_{Y_1,Y_2|X}$};
\draw [->,thick] (5.0, -0.8) --(6.0, 0); \node at (5.5,0) {$Y_1^n$};
\draw [->,thick] (5.0, -0.8) --(6.0, -1.6); \node at (5.5,-1.6) {$Y_2^n$};
\draw (6.0, 0.4) rectangle +(1.6, -0.8); \node at (6.8,0) {Relay Node $1$};
\draw (6.0, -1.2) rectangle +(1.6, -0.8); \node at (6.8,-1.6) {Relay Node $2$};
\draw [->,thick] (7.6, 0) --(8.8,-0.8); \node at (8.2, 0) {$J_{1,n}/C_1$};
\draw [->,thick] (7.6, -1.6) --(8.8,-0.8); \node at (8.2, -1.6) {$J_{2,n}/C_2$};
\draw (8.8, -0.4) rectangle +(1.6,-0.8); \node at (9.6,-0.75) {Receiver};
\end{tikzpicture}
    \caption{Primitive diamond relay communication system.}
    \label{fig:1}
\end{figure}
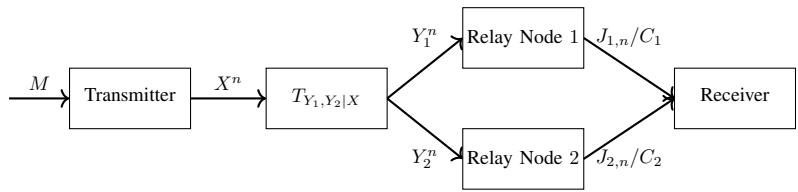
For the scalar Gaussian primitive diamond relay channel, the source-to-relay broadcast channel is
\begin{align*}
        Y_1 =& X+Z_1,\\
        Y_2 =& X+Z_2,
\end{align*}
where
$$Z_1\sim \mathcal{N}(0,N_1),\qquad Z_2\sim \mathcal{N}(0,N_2).$$
The noises $Z_1$ and $Z_2$ are mutually independent and are independent of $X$. The channel input is subject to the average power constraint $\mathbb E[X^2]\leq P$.

\subsection{Upper and Lower Bounds on the Capacity}
The capacity of the primitive diamond relay channel remains unknown in general.
The cut-set bound provides a natural upper bound.
\begin{theorem}{(Cut-set Bound)}
    Any achievable rate $R$ for a primitive diamond relay channel must satisfy
    \begin{subequations}
        \begin{align}
            R\leq &I(X;Y_1,Y_2),\\
            R\leq &C_1+I(X;Y_2),\\
            R\leq &C_2+I(X;Y_1),\\
            R\leq &C_1+C_2,
        \end{align}
    \end{subequations}
    for some input distribution $p(x)$.
\end{theorem}
The cut-set bound is known to be not tight in general, since the channel setting permits unrestricted cooperation among the terminals on either side of a cut. Hence, a stronger converse should account more carefully for the dependence among the source input, the relay observations, and the descriptions transmitted by the relays.

The primitive diamond relay channel can be viewed as a two-relay extension of the primitive relay channel. For the latter in the Gaussian case, Wu, Barnes, and Özgür developed a converse based on a new isoperimetric inequality~\cite{wbo19,wbo17}; a simple proof for the Gaussian case based on information constrained optimal-transport was subsequently given in~\cite{bwo23}. Building on this approach, Wu, Özgür, Peleg, and Shamai derived upper bounds for both Gaussian and discrete memoryless primitive diamond relay channels that improve upon the cut-set bound~\cite{wu2019new}.

\begin{theorem}\cite{wu2019new}[Theorem 3.1]\label{ob:ayfergaussian}
    For a Gaussian primitive diamond relay channel, if a rate $R$ is achievable, then there exists $\theta_i\in[\arcsin(2^{-C_i}),\frac{\pi}{2}],i\in\{1,2\}$ such that
    \begin{subequations}
        \begin{align}
            R\leq& \frac{1}{2}\log(1+\frac{P}{N_1})+C_2+\log\sin\theta_2,\label{eqn:ayfergauss1}\\
            R\leq& \frac{1}{2}\log(1+\frac{P}{N_2})+\min_{\omega_1\in (\frac{\pi}{2}-\theta_1,\frac{\pi}{2}]}h_1(\omega_1;\theta_1),\label{eqn:ayfergauss2}\\
            R\leq& \frac{1}{2}\log(1+\frac{P}{N_2})+C_1+\log\sin\theta_1,\label{eqn:ayfergauss3}\\
            R\leq& \frac{1}{2}\log(1+\frac{P}{N_1})+\min_{\omega_2\in (\frac{\pi}{2}-\theta_2,\frac{\pi}{2}]}h_2(\omega_2;\theta_2),\label{eqn:ayfergauss4}\\
            R\leq & C_1+C_2+\log \sin\theta_1+\log\sin\theta_2,\label{eqn:ayfergauss5}
        \end{align}
        where 
        \begin{align}
            h_i(\omega;\theta):=\frac{1}{2}\log\bigg(\frac{[P(N_1+N_2)+N_1N_2\sin^2\omega -2P\sqrt{N_1 N_2}\cos \omega]\sin^2\theta}{N_i(P+N_{\bar i})(\sin^2\theta -\cos^2\omega)}\bigg).
        \end{align}
    \end{subequations}
\end{theorem}

\begin{theorem}\cite{wu2019new}[Theorem 3.2]\label{ob:ayferdiscrete}
For a discrete memoryless primitive diamond relay channel, if a rate $R$ is achievable, then there exists some $p(x)$ and $a_1,a_2\geq 0$ such that
\begin{subequations}
    \begin{align}
        R\leq &I(X;Y_1,Y_2),\\
        R\leq &I(X;Y_1)+C_2-a_2,\\
        R\leq &I(X;Y_1,\tilde Y_2)+H_2(\sqrt{a_2\frac{\ln 2}{2}})+\sqrt{a_2\frac{\ln 2}{2}}\log(|\mathcal{Y}_2|-1)-a_2,\\
        R\leq & I(X;Y_2)+C_1-a_1,\\
        R\leq & I(X;\tilde Y_1,Y_2)+H_2(\sqrt{a_1\frac{\ln 2}{2}})+\sqrt{a_1\frac{\ln 2}{2}}\log(|\mathcal{Y}_1|-1)-a_1,\\
        R\leq & C_1+C_2-a_1-a_2,
    \end{align}
\end{subequations}
for any random variables $\tilde Y_1$ and $\tilde Y_2$ with the same conditional distributions as $Y_1$ and $Y_2$ given $X$.
\end{theorem}

El Gamal, Gohari, and Nair developed a strengthened cut-set upper bound on the capacity of the primitive relay channel~\cite{egn21}. Their converse employs the Csisz\'ar--K\"orner--Marton sum identity and Gallager-type identification to obtain a single-letter characterization. In the Gaussian setting, their bound coincides with the upper bound of~\cite{wbo17}, whereas in the discrete setting it can be strictly tighter. By adapting their converse method to the two-relay model, we obtain a new upper bound for the primitive diamond relay channel, which will be presented in Section~\ref{sec:proofnewob}. 

To assess the tightness of the upper bounds, we consider a standard achievable benchmark obtained by coded time sharing between decode-and-forward and distributed compress-and-forward/noisy network coding~\cite{CoverElGamal1979,LimKimElGamalChung2011NNC}. The scheme jointly optimizes the time-sharing fraction and the allocation of the relay-link capacities between the two transmission modes. In the Gaussian setting, we evaluate the scheme using Gaussian inputs and independent Gaussian quantization test channels.

\begin{theorem}[DF--CF time-sharing inner bound]
\label{thm:df-cf-timesharing}
Consider a primitive diamond relay channel,  every rate $R<r_{\mathrm{D}}+r_{\mathrm{C}}$ is achievable for some nonnegative $r_{\mathrm{D}}$ and
$r_{\mathrm{C}}$ satisfying
\begin{subequations}
\label{eq:df-cf-timesharing}
\begin{align}
    r_{\mathrm{D}}
    &\leq
    \lambda I_{\mathrm{D}}(X;Y_1),
    \label{eq:df-cf-df1}\\
    r_{\mathrm{D}}
    &\leq
    \lambda I_{\mathrm{D}}(X;Y_2),
    \label{eq:df-cf-df2}\\
    r_{\mathrm{D}}
    &\leq
    a_1+a_2,
    \label{eq:df-cf-df3}\\
    r_{\mathrm{C}}
    &\leq
    (1-\lambda)
    I_{\mathrm{C}}(X;\hat Y_1,\hat Y_2),
    \label{eq:df-cf-cf1}\\
    r_{\mathrm{C}}
    &\leq
    C_1-a_1
    +(1-\lambda)
    \bigl[
        I_{\mathrm{C}}(X;\hat Y_2)
        -I_{\mathrm{C}}(Y_1;\hat Y_1|X)
    \bigr],
    \label{eq:df-cf-cf2}\\
    r_{\mathrm{C}}
    &\leq
    C_2-a_2
    +(1-\lambda)
    \bigl[
        I_{\mathrm{C}}(X;\hat Y_1)
        -I_{\mathrm{C}}(Y_2;\hat Y_2|X)
    \bigr],
    \label{eq:df-cf-cf3}\\
    r_{\mathrm{C}}
    &\leq
    C_1+C_2-a_1-a_2
    -(1-\lambda)
    \bigl[
        I_{\mathrm{C}}(Y_1;\hat Y_1|X)
        +I_{\mathrm{C}}(Y_2;\hat Y_2|X)
    \bigr].
    \label{eq:df-cf-cf4}
\end{align}
\end{subequations}
Here, $I_{\mathrm{D}}$ denotes mutual information evaluated
according to a distribution of the form
\[
    p_{\mathrm{D}}(x,y_1,y_2)
    =
    p_{\mathrm{D}}(x)p(y_1|x)p(y_2|x),
\]
whereas $I_{\mathrm{C}}$ denotes mutual information evaluated
according to a distribution of the form
\[
    p_{\mathrm{C}}(x,y_1,y_2,\hat y_1,\hat y_2) =
    p_{\mathrm{C}}(x)p(y_1|x)p(y_2|x)
    p(\hat y_1|y_1)p(\hat y_2|y_2).
\]
It suffices to consider compression alphabets satisfying $|\mathcal{\hat Y}_1|\leq |\mathcal{Y}_1|+3$ and $|\mathcal{\hat Y}_2|\leq |\mathcal{Y}_2|+3$.

The union is taken over all such distributions and all
\[
    \lambda\in[0,1],\qquad
    0\leq a_1\leq C_1,\qquad
    0\leq a_2\leq C_2.
\]
\end{theorem}

The remainder of the paper is organized as follows. In Section~\ref{sec:proofnewob}, we present and prove the proposed upper bound. In Section~\ref{sec:gauss}, we establish the equivalence between the Gaussian specialization of our proposed bound and the upper bound in Theorem~\ref{ob:ayfergaussian}. In Section~\ref{sec:discrete}, we show that our proposed bound can be strictly tighter than the discrete-memoryless bound in Theorem~\ref{ob:ayferdiscrete}. Finally, Section~\ref{conclusion} concludes the paper and discusses directions for future work.

\section{The New Upper Bound}\label{sec:proofnewob}
\begin{theorem}\label{ob:relayuvw}
    Any achievable rate $R$ for a primitive diamond relay channel must satisfy
    \begin{subequations}
        \begin{align}
        R\leq &I(X;Y_2)+C_1-I(U;Y_1|X),\label{eqn:obuvw1}\\
        R\leq &I(X;Y_1)+I(X;Y_2|U)-I(X;Y_1|U),\label{eqn:obuvw2}\\
        R\leq & I(X;Y_1)+C_2-I(V;Y_2|X),\label{eqn:obuvw3}\\
        R\leq & I(X;Y_2)+I(X;Y_1|V)-I(X;Y_2|V),\label{eqn:obuvw4}\\
        R\leq &C_1+C_2-I(U;Y_1|X)-I(V;Y_2|X).\label{eqn:obuvw5}
    \end{align}
    \end{subequations}
    The union is taken over distributions of the form $p(x)p(y_1|x)p(y_2|x)p(u|x,y_1)p(v|x,y_2)$, so that $U-(X,Y_1)-Y_2$ and $V-(X,Y_2)-Y_1$.
It suffices to consider auxiliaries satisfying $|U|\leq |X||Y_1|+1$, $|V|\leq |X||Y_2|+1$.
\end{theorem}

\begin{proof}
The argument follows in the same flavor to the proof of \cite[Proposition~1]{egn21}, with the modifications needed to accommodate the second relay link. We consider the joint distribution $(M,X^n,Y_1^n,Y_2^n,J_{1,n},J_{2,n})$ induced by any given code.

We start with Fano's inequality,
$$nR = H(M)\leq I(M;J_{1,n},J_{2,n})+n\epsilon_n.$$

We identify $U_i =(Y_{2,i+1}^n,Y_1^{i-1},J_{1,n})$, and $V_i = (Y_{2,i+1}^n,Y_1^{i-1},J_{2,n})$. Let $Q$ be a time-sharing random variable uniform on $[1:n]$ and independent of everything else. Further let $U =(Q,U_Q),V=(Q,V_Q), X = X_Q, Y_1 = Y_{1,Q},Y_2=Y_{2,Q}$. 

\smallskip
\noindent\emph{First bound.}
\begin{align*}
    n(R-\epsilon_n) \leq &I(M;J_{1,n},J_{2,n})\leq  I(X^n;J_{1,n},Y_2^n)\\
    =& I(X^n;J_{1,n},Y_2^n)+I(X^n;Y_1^n,J_{1,n})-I(X^n;Y_1^n,J_{1,n})\\
    \overset{(a)}{=}&I(X^n;Y_1^n)+I(X^n;Y_2^n|J_{1,n})-I(X^n;Y_1^n|J_{1,n})\\
    \overset{(b)}{=}&
    \sum_{i=1}^n I(X_i;Y_{1,i}|Y^{i-1})
    +I(X_i;Y_{2,i}|J_{1,n},Y_1^{i-1},Y_{2,i+1}^n)
    -I(X_i;Y_{1,i}|J_{1,n},Y_1^{i-1},Y_{2,i+1}^n)\\
    \overset{(c)}{\leq} &\sum_{i=1}^n I(X_i;Y_{1,i})
    +I(X_i;Y_{2,i}|J_{1,n},Y_1^{i-1},Y_{2,i+1}^n)
    -I(X_i;Y_{1,i}|J_{1,n},Y_1^{i-1},Y_{2,i+1}^n)\\
    =&\sum_{i=1}^n I(X_i;Y_{1,i})
    +I(X_i;Y_{2,i}|U_i)
    -I(X_i;Y_{1,i}|U_i)\\
    = & n\bigg(I(X_Q;Y_{1,Q}|Q)
    +I(X_Q;Y_{2,Q}|U_Q,Q)
    -I(X_Q;Y_{1,Q}|U_Q,Q)\bigg)\\
    \overset{(d)}{\leq}&n\bigg(I(X;Y_1)+I(X;Y_2|U)-I(X;Y_1|U)\bigg).
\end{align*}
Here $(a)$ uses $J_{1,n}= f_{1,n}(Y_1^n)$; $(b)$ follows by the telescoping sum and the Csisz\'ar--K\"orner--Marton sum identity; $(c)$ drops $-I(Y_i;Y^{i-1})$; $(d)$ holds since $Q\rightarrow X_Q\rightarrow (Y_{1,Q},Y_{2,Q})$ is Markov.

\smallskip
\noindent\emph{Second bound.}
\begin{align*}
    n (R-\epsilon_n)\leq & I(M;J_{1,n},J_{2,n})\leq I(X^n;J_{1,n},Y_2^n)\\
    \overset{(a)}{=}& I(X^n;Y_2^n)+I(X^n;J_{1,n}|Y_2^n)+H(J_{1,n}|Y_1^n,X^n,Y_2^n)\\
    =&I(X^n;Y_2^n)-I(J_{1,n};Y_1^n|X^n,Y_2^n)+H(J_{1,n}|Y_2^n)\\
    \overset{(b)}{\leq} &nC_1+ I(X^n;Y_2^n)-I(J_{1,n};Y_1^n|X^n)\\
    \overset{(c)}{=}&nC_1+\sum_{i=1}^n I(X_i;Y_{2,i}|Y_{2,i+1}^n)-I(J_{1,n};Y_{1,i}|X^n,Y_1^{i-1})\\
    \overset{(d)}{\leq } &nC_1+\sum_{i=1}^n I(X_i;Y_{2,i})-I(J_{1,n};Y_{1,i}|X^n,Y_1^{i-1})\\
    \overset{(e)}{\leq}&nC_1+\sum_{i=1}^n I(X_i;Y_{2,i})-I(J_{1,n},Y_1^{i-1},Y_{2,i+1}^n;Y_{1,i}|X_i) 
    \\
    =&nC_1+\sum_{i=1}^n I(X_i;Y_{2,i})-I(U_i;Y_{1,i}|X_i) 
    \\
     =&nC_1+n\bigg( I(X_Q;Y_{2,Q}|Q)-I(U_Q;Y_{1,Q}|X_Q,Q)\bigg)\\
    \overset{(f)}{\leq}&nC_1+n\bigg( I(X;Y_2)-I(U;Y_1|X)\bigg).
\end{align*}
Here $(a)$ again uses $J_{1,n} =f_{1,n}(Y_1^n)$; $(b)$ holds by $H(J_{1,n}|Y_2^n)\leq H(J_{1,n})\leq C_1$; $(c)$ follows by the Csisz\'ar--K\"orner--Marton sum identity; $(d)$ drops $-I(Y_{2,i};Y_{2,i+1}^n)$; $(e)$ follows by the Markov chain conditions $(Y_1^{i-1},X^{n\setminus i})\rightarrow X_i\rightarrow Y_{1,i}$, $Y_{2,i+1}^n\rightarrow (J_{1,n},Y_1^{i-1},X^n)\rightarrow Y_{1,i}$, and dropping $I(X^{n\setminus i}; Y_{1,i}|J_{1,n},Y_1^{i-1},Y_{2,i+1}^n,X_i)$; $(f)$ holds since $Q\rightarrow X_Q\rightarrow (Y_{1,Q},Y_{2,Q})$ is Markov.
 
\smallskip
\noindent\emph{Third and fourth bounds.}
Exchanging the roles of the two relays (i.e.\ interchanging the indices $1\leftrightarrow2$,
replacing $U$ by $V$) in the two chains above gives
\begin{align*}
    n(R-\epsilon_n)\leq & n\bigg(I(X;Y_2)+I(X;Y_1|V)-I(X;Y_2|V)\bigg),\\
    n(R-\epsilon_n)\leq& nC_2+n\bigg(I(X;Y_1)-I(V;Y_2|X)\bigg).
\end{align*}
\smallskip
\noindent\emph{Fifth bound.}
\begin{align*}
    n(R-\epsilon_n)\leq& I(M;J_{1,n},J_{2,n}) \leq I(X^n;J_{1,n},J_{2,n})\\
    \leq& I(X^n;J_{1,n})+I(X^n;J_{2,n})\\
    \overset{(a)}{\leq} & nC_1+nC_2-H(J_{1,n}|X^n)-H(J_{2,n}|X^n)\\
    \overset{(b)}{=} &nC_1+nC_2-I(J_{1,n};Y_1^n|X^n)-I(J_{2,n};Y_2^n|X^n)\\
    \overset{(c)}{=}&nC_1+nC_2+\sum_{i=1}^n-I(J_{1,n};Y_{1,i}|X^n,Y_1^{i-1})-I(J_{2,n};Y_{2,i}|X^n,Y_{2,i+1}^n)\\
    \overset{(d)}{\leq }&nC_1+nC_2+\sum_{i=1}^n-I(J_{1,n},Y_1^{i-1},Y_{2,i+1}^n;Y_{1,i}|X_i)-I(J_{2,n},Y_1^{i-1},Y_{2,i+1}^n;Y_{2,i}|X_i)\\
    =&nC_1+nC_2+\sum_{i=1}^n-I(U_i;Y_{1,i}|X_i)-I(V_i;Y_{2,i}|X_i)\\
    =&nC_1+nC_2+n\bigg(-I(U_Q;Y_{1,Q}|X_Q,Q)-I(V_Q;Y_{2,Q}|X_Q,Q)\bigg)\\
    \overset{(e)}{=}&n\bigg(C_1+C_2-I(U;Y_1|X)-I(V;Y_2|X)\bigg).
\end{align*}
Here $(a)$ follows by $H(J_{1,n})\leq C_1$ and $H(J_{2,n})\leq C_2$; $(b)$ follows as $J_{1,n} = f_{1,n}(Y_1^n)$ and $J_{2,n} = f_{2,n}(Y_2^n)$; $(c)$ uses the telescoping sum; $(d)$ follows by the Markov chain conditions $(X^{n\setminus i},Y_1^{i-1})\rightarrow X_i\rightarrow Y_{1,i}$, $Y_{2,i+1}^n\rightarrow (J_{1,n},Y_1^{i-1},X^n)\rightarrow Y_{1,i}$, $(X^{n\setminus i},Y_{2,i+1}^n)\rightarrow X_i\rightarrow Y_{2,i}$, $Y_1^{i-1}\rightarrow (J_{2,n},Y_{2,i+1}^n,X^n)\rightarrow Y_{2,i}$, dropping $-I(X^{n\setminus i}; Y_{1,i}|J_{1,n},Y_1^{i-1},Y_{2,i+1}^n,X_i)$, and dropping $-I(X^{n\setminus i}; Y_{2,i}|J_{2,n},Y_1^{i-1},Y_{2,i+1}^n,X_i)$; $(e)$ holds as $Q\rightarrow X_Q\rightarrow (Y_{1,Q},Y_{2,Q})$ is Markov. 

Letting $n\to\infty$ so that $\epsilon_n\to0$, and invoking continuity of the
mutual-information terms in the induced joint distribution, gives the five claimed constraints.

The cardinality bounds $|U|\leq |X||Y_1|+1$ and $|V|\leq |X||Y_2|+1$ are due to the standard Caratheodory-Bunt\cite{bunt1934bijdrage} arguments.
\end{proof}

\begin{remark}
After some elementary manipulation, the five constraints in Theorem~\ref{ob:relayuvw} can be rewritten as
\begin{align*}
        R\leq &I(X;Y_2)+C_1-I(U;Y_1|X),\\
        R\leq &I(X;Y_2,Y_1)-I(U;Y_2|Y_1)-I(X;Y_1|U,Y_2),\\
        R\leq & I(X;Y_1)+C_2-I(V;Y_2|X),\\
        R\leq & I(X;Y_2,Y_1)-I(V;Y_1|Y_2)-I(X;Y_2|V,Y_1),\\
        R\leq &C_1+C_2-I(U;Y_1|X)-I(V;Y_2|X).
    \end{align*}
By non-negativity of mutual information terms, the upper bound in Theorem~\ref{ob:relayuvw} is at least as tight as the cut-set bound for the general primitive diamond relay channel.
\end{remark}
\begin{remark}\label{rm:constraint_C}
    When evaluating Theorem~\ref{ob:relayuvw}, it suffices to consider the random variables satisfying
    \begin{align*}
        C_1&\geq I(U;Y_1|X),\\
        C_2&\geq I(V;Y_2|X).
    \end{align*}
    The argument is contained in the converse proof. For instance, we have
    \begin{align*}
        0\leq &I(X^n;J_{1,n}|Y_2^n)+H(J_{1,n}|Y_1^n,X^n,Y_2^n)\\
        \leq &n(C_1-I(U;Y_1|X)).
    \end{align*}
\end{remark}

\section{Gaussian Primitive Diamond Relay Channel}\label{sec:gauss}
 
\begin{lemma}\label{lm:gaussianopt}
    To evaluate the upper bound in Theorem~\ref{ob:relayuvw} in the Gaussian setting, it suffices to consider the Gaussian distributions $P_{X}$, $P_{X,Y_1|U}$, and $P_{X,Y_2|V}$.
\end{lemma}
The Gaussian optimality proof can be proved in a similar flavor as in \cite{gn14,gohari2021outer}, using the standard rotation trick. See Appendix~\ref{app:gaussianopt} for the detailed proof.

\begin{proposition}\label{prop:obuvwgauss}
For a Gaussian primitive diamond relay channel, any achievable rate $R$ must satisfy
    \begin{subequations}
    \begin{align}
        R\leq &\tfrac{1}{2}\log\Bigl(1+\tfrac{P}{N_2}\Bigr)+C_1
               +\tfrac{1}{2}\log\Bigl[\tfrac{K_2}{N_1}\bigl(1-\rho^2\bigr)\Bigr],
               \label{eqn:snrbound1}\\
        R\leq &\tfrac{1}{2}\log\Bigl(1+\tfrac{P}{N_1}\Bigr)
               +\tfrac{1}{2}\log\Bigl[\Bigl(1+\tfrac{K_1}{N_2}\Bigr)
                \bigl(1-\rho^2\bigr)\Bigr],\label{eqn:snrbound2}\\
        R\leq &\tfrac{1}{2}\log\Bigl(1+\tfrac{P}{N_1}\Bigr)+C_2
               +\tfrac{1}{2}\log\Bigl[\tfrac{\hat K_2}{N_2}
                \bigl(1-\hat\rho^2\bigr)\Bigr],\label{eqn:snrbound3}\\
        R\leq &\tfrac{1}{2}\log\Bigl(1+\tfrac{P}{N_2}\Bigr)
               +\tfrac{1}{2}\log\Bigl[\Bigl(1+\tfrac{\hat K_1}{N_1}\Bigr)
                \bigl(1-\hat\rho^2\bigr)\Bigr],\label{eqn:snrbound4}\\
        R\leq& C_1+C_2
               +\tfrac{1}{2}\log\Bigl[\tfrac{K_2}{N_1}\bigl(1-\rho^2\bigr)\Bigr]
               +\tfrac{1}{2}\log\Bigl[\tfrac{\hat K_2}{N_2}
                \bigl(1-\hat\rho^2\bigr)\Bigr],\label{eqn:snrbound5}
    \end{align}
    for some $0\leq K_1,\hat K_1\leq P$, $0\leq K_2\leq P+N_1$, $0\leq \hat K_2\leq P+N_2$, and $\rho,\hat \rho\in[0,1]$,
    subject to
    \begin{align}
        (P-K_1)(P+N_1-K_2)\geq& \bigl(P-\rho\sqrt{K_1 K_2}\bigr)^2,
        \label{eqn:snrbound6}\\
        (P-\hat K_1)(P+N_2-\hat K_2)\geq&
        \bigl(P-\hat \rho\sqrt{\hat K_1 \hat K_2}\bigr)^2,\label{eqn:snrbound7}\\
        C_1+\tfrac{1}{2}\log\Bigl[\tfrac{K_2}{N_1}\bigl(1-\rho^2\bigr)\Bigr]
        \geq& 0,\label{eqn:snrbound8}\\
        C_2+\tfrac{1}{2}\log\Bigl[\tfrac{\hat K_2}{N_2}
        \bigl(1-\hat\rho^2\bigr)\Bigr]\geq & 0,\label{eqn:snrbound9}
    \end{align}
\end{subequations}
\end{proposition}
\begin{proof}
    By Lemma~\ref{lm:gaussianopt}, when evaluating Theorem~\ref{ob:relayuvw}, we consider the Gaussian distributions $P_X$, $P_{X,Y_1|U}$, and $P_{X,Y_2|V}$. Accordingly, let the variance of $X$ be $P_0$, and let the conditional covariances of $X,Y_1$ given $U$ and of $X,Y_2$ given $V$ be 
\begin{subequations}
\begin{align}
    K_{X,Y_1|U} =&\begin{bmatrix}
        K_1 &\rho \sqrt{K_1 K_2}\\
        \rho \sqrt{K_1 K_2} & K_2
    \end{bmatrix}\preceq \begin{bmatrix}
        P_0 & P_0\\ P_0 & P_0+N_1
    \end{bmatrix},\label{eqn:psdxyu}\\
    K_{X,Y_2|V} =&\begin{bmatrix}
        \hat K_1 &\hat \rho \sqrt{\hat K_1 \hat K_2}\\
        \hat \rho \sqrt{\hat K_1 \hat K_2} & \hat K_2
    \end{bmatrix}\preceq \begin{bmatrix}
        P_0 & P_0\\ P_0 & P_0+N_2
    \end{bmatrix},\label{eqn:psdxyv}
\end{align}
\end{subequations}
for some $0\leq K_1,\hat K_1\leq P_0$, $0\leq K_2\leq P_0+N_1$, $0\leq \hat K_2\leq P_0+N_2$, and $\rho,\hat \rho\in[-1,1]$. 

With this parameterization, it is optimal to set $P_0 = P$. Hence, \eqref{eqn:obuvw1}--\eqref{eqn:obuvw5} specializes to \eqref{eqn:snrbound1}--\eqref{eqn:snrbound5}. Moreover, \eqref{eqn:snrbound6} and \eqref{eqn:snrbound7} are precisely the positive semi-definiteness requirements in \eqref{eqn:psdxyu} and \eqref{eqn:psdxyv}, whereas \eqref{eqn:snrbound8} and \eqref{eqn:snrbound9} follow from Remark~\ref{rm:constraint_C}.

Lastly, the bounds \eqref{eqn:snrbound1}--\eqref{eqn:snrbound5} depend on
$\rho,\hat\rho$ only through $\rho^2,\hat\rho^2$, while for $\rho<0$
\[
    \bigl(P-|\rho|\sqrt{K_1K_2}\bigr)^2
    \leq\bigl(P-\rho\sqrt{K_1K_2}\bigr)^2
    \leq (P-K_1)(P+N_1-K_2),
\]
so that $(K_1,K_2,|\rho|)$ remains feasible and attains the same five values;
the same argument applies to $\hat\rho$.  We may therefore restrict attention
to $\rho,\hat\rho\in[0,1]$.
\end{proof}

\begin{lemma}\label{lm:gaussianequivalence}
    For the Gaussian primitive diamond relay channel, the upper bounds in Theorem~\ref{ob:ayfergaussian} and Theorem~\ref{ob:relayuvw} coincide.
\end{lemma}
\begin{proof}
It suffices to show the upper bounds in Proposition~\ref{prop:obuvwgauss} and Theorem~\ref{ob:ayfergaussian} coincide.

\emph{Step 1: elimination of $\rho$ and $\hat\rho$.}

Each right-hand side in \eqref{eqn:snrbound1}--\eqref{eqn:snrbound5} is
non-increasing in $\rho^2$ and in $\hat\rho^2$, so it is optimal to choose
$\rho^2,\hat\rho^2$ as small as permitted:
\begin{subequations}
    \begin{align}
        \rho = &\max\bigg\{0,\frac{P-\sqrt{(P-K_1)(P+N_1-K_2)}}{\sqrt{K_1K_2}}\bigg\},\label{eqn:k1k2rho1}\\
        \hat \rho=&\max\bigg\{0, \frac{P-\sqrt{(P-\hat K_1)(P+N_2-\hat K_2)}}{\sqrt{\hat K_1\hat K_2}}\bigg\}.\label{eqn:k1k2rho2}
    \end{align}
\end{subequations}
When $\rho,\hat \rho>0$, \eqref{eqn:k1k2rho1} and \eqref{eqn:k1k2rho2} can be rewritten as
\begin{subequations}\label{eqn:sident}
    \begin{align}
        K_2\bigl(1-\rho^2\bigr) =& N_1
        -\frac{P}{K_1}\Bigl(\sqrt{P+N_1-K_2}-\sqrt{P-K_1}\Bigr)^2
        \ \leq\ N_1,\\
        \hat K_2\bigl(1-\hat\rho^2\bigr) =& N_2
        -\frac{P}{\hat K_1}\Bigl(\sqrt{P+N_2-\hat K_2}
        -\sqrt{P-\hat K_1}\Bigr)^2\ \leq\ N_2 .
    \end{align}
\end{subequations}

\emph{Step 2: trigonometric substitution.}

In view of \eqref{eqn:snrbound8}, \eqref{eqn:snrbound9}, and \eqref{eqn:sident}, define
\begin{subequations}\label{eqn:chart}
    \begin{align}
        \psi_1:=\arcsin\sqrt{\frac{K_2\bigl(1-\rho^2\bigr)}{N_1}}
        \ \in\ \Bigl[\arcsin\bigl(2^{-C_1}\bigr),\tfrac{\pi}{2}\Bigr],\\
        \psi_2:=\arcsin\sqrt{\frac{\hat K_2\bigl(1-\hat\rho^2\bigr)}{N_2}}
        \ \in\ \Bigl[\arcsin\bigl(2^{-C_2}\bigr),\tfrac{\pi}{2}\Bigr].
    \end{align}
\end{subequations}
After substitution, \eqref{eqn:snrbound1}--\eqref{eqn:snrbound9} then equal
\begin{subequations}
    \begin{align}
        R\leq &\tfrac{1}{2}\log\Bigl(1+\tfrac{P}{N_2}\Bigr)+C_1
               +\log\sin\psi_1,\label{eqn:psibound1}\\
        R\leq &\tfrac{1}{2}\log\Bigl(1+\tfrac{P}{N_1}\Bigr)
               +\tfrac{1}{2}\log\Bigl(1+\tfrac{K_1}{N_2}\Bigr)
               +\tfrac{1}{2}\log\frac{N_1\sin^{2}\psi_1}
               {P+N_1-\bigl(\sqrt{P-K_1}
                +\cos\psi_1\sqrt{\tfrac{N_1K_1}{P}}\bigr)^{2}},
               \label{eqn:psibound2}\\
        R\leq &\tfrac{1}{2}\log\Bigl(1+\tfrac{P}{N_1}\Bigr)+C_2
               +\log\sin\psi_2,\label{eqn:psibound3}\\
        R\leq &\tfrac{1}{2}\log\Bigl(1+\tfrac{P}{N_2}\Bigr)
               +\tfrac{1}{2}\log\Bigl(1+\tfrac{\hat K_1}{N_1}\Bigr)
               +\tfrac{1}{2}\log\frac{N_2\sin^{2}\psi_2}
               {P+N_2-\bigl(\sqrt{P-\hat K_1}
                +\cos\psi_2\sqrt{\tfrac{N_2\hat K_1}{P}}\bigr)^{2}},
               \label{eqn:psibound4}\\
        R\leq & C_1+C_2+\log\sin\psi_1+\log\sin\psi_2 ,\label{eqn:psibound5}
    \end{align}
\end{subequations}
    for some $\psi_1\in(\arcsin(2^{-C_1}),\frac{\pi}{2}]$, $\psi_2\in(\arcsin(2^{-C_2}),\frac{\pi}{2}]$, and $K_1\in[\frac{PN_1\cos^2\psi_1}{P+N_1\cos^2\psi_1},\,P]$, $\hat K_1\in[\frac{PN_2\cos^2\psi_2}{P+N_2\cos^2\psi_2},\,P]$.

Upon relating $(\psi_1,\psi_2)$ with $(\theta_1,\theta_2)$, the constraints \eqref{eqn:psibound1}, \eqref{eqn:psibound3}, and \eqref{eqn:psibound5} coincide with \eqref{eqn:ayfergauss1}, \eqref{eqn:ayfergauss3}, and \eqref{eqn:ayfergauss5}, respectively. It remains to show that \eqref{eqn:psibound2} and \eqref{eqn:psibound4} coincide with \eqref{eqn:ayfergauss2} and \eqref{eqn:ayfergauss4}.

\emph{Step 3: two explicit optimizations.}

A direct computation yields
\begin{subequations}
    \begin{align}
    &\max_{K_1\in\bigl[\frac{PN_1\cos^{2}\psi_1}{P+N_1\cos^{2}\psi_1},\,P\bigr]}
    \frac{1}{2}\log\Bigl(1+\frac{P}{N_1}\Bigr)
    +\frac{1}{2}\log\Bigl(1+\frac{K_1}{N_2}\Bigr)
    +\frac{1}{2}\log\frac{N_1\sin^{2}\psi_1}
      {P+N_1-\bigl(\sqrt{P-K_1}
       +\cos\psi_1\sqrt{\frac{N_1K_1}{P}}\bigr)^{2}}\nonumber\\
    =&\frac{1}{2}\log\left[\left(1+
      \frac{P(N_1+N_2)+N_1N_2\cos^{2}\psi_1
            +\sqrt{\bigl[P(N_1+N_2)+N_1N_2\cos^{2}\psi_1\bigr]^{2}
                   -4P^{2}N_1N_2\sin^{2}\psi_1}}
           {2N_1N_2\sin^{2}\psi_1}\right)\sin^{2}\psi_1\right],
    \label{eqn:finalpsibound2}\\
    &\max_{\hat K_1\in\bigl[\frac{PN_2\cos^{2}\psi_2}{P+N_2\cos^{2}\psi_2},\,P\bigr]}
    \frac{1}{2}\log\Bigl(1+\frac{P}{N_2}\Bigr)
    +\frac{1}{2}\log\Bigl(1+\frac{\hat K_1}{N_1}\Bigr)
    +\frac{1}{2}\log\frac{N_2\sin^{2}\psi_2}
      {P+N_2-\bigl(\sqrt{P-\hat K_1}
       +\cos\psi_2\sqrt{\frac{N_2\hat K_1}{P}}\bigr)^{2}}\nonumber\\
    =&\frac{1}{2}\log\left[\left(1+
      \frac{P(N_1+N_2)+N_1N_2\cos^{2}\psi_2
            +\sqrt{\bigl[P(N_1+N_2)+N_1N_2\cos^{2}\psi_2\bigr]^{2}
                   -4P^{2}N_1N_2\sin^{2}\psi_2}}
           {2N_1N_2\sin^{2}\psi_2}\right)\sin^{2}\psi_2\right],
    \label{eqn:finalpsibound4}
\end{align}
\end{subequations}
where the maximizers, whose feasibility proof is omitted here, are given by
\begin{align*}
    K_1^{\star}&=\frac{P^{2}\bigl[P(N_1-N_2)-N_1N_2\cos^{2}\psi_1+A\bigr]^{2}}
    {P\bigl[P(N_1-N_2)-N_1N_2\cos^{2}\psi_1+A\bigr]^{2}
     +N_1\cos^{2}\psi_1
      \bigl[2P^{2}+P(N_1+N_2)+N_1N_2\cos^{2}\psi_1-A\bigr]^{2}},\\
    A&:=\sqrt{\bigl[P(N_1+N_2)+N_1N_2\cos^{2}\psi_1\bigr]^{2}
             -4P^{2}N_1N_2\sin^{2}\psi_1},\\
    \hat K_1^{\star}&=\frac{P^{2}\bigl[P(N_2-N_1)-N_1N_2\cos^{2}\psi_2
        +\hat A\bigr]^{2}}
    {P\bigl[P(N_2-N_1)-N_1N_2\cos^{2}\psi_2+\hat A\bigr]^{2}
     +N_2\cos^{2}\psi_2
      \bigl[2P^{2}+P(N_1+N_2)+N_1N_2\cos^{2}\psi_2-\hat A\bigr]^{2}},\\
    \hat A&:=\sqrt{\bigl[P(N_1+N_2)+N_1N_2\cos^{2}\psi_2\bigr]^{2}
             -4P^{2}N_1N_2\sin^{2}\psi_2}.
\end{align*}

The minimization problems in \eqref{eqn:ayfergauss2} and
\eqref{eqn:ayfergauss4} admit explicit solutions as well:
\begin{subequations}
    \begin{align}
    &\frac{1}{2}\log\Bigl(1+\frac{P}{N_2}\Bigr)
    +\min_{\omega_1\in(\frac{\pi}{2}-\theta_1,\frac{\pi}{2}]}
     h_1(\omega_1;\theta_1)\nonumber\\
    =&\frac{1}{2}\log\left[\left(1+
      \frac{P(N_1+N_2)+N_1N_2\cos^{2}\theta_1
            +\sqrt{\bigl[P(N_1+N_2)+N_1N_2\cos^{2}\theta_1\bigr]^{2}
                   -4P^{2}N_1N_2\sin^{2}\theta_1}}
           {2N_1N_2\sin^{2}\theta_1}\right)\sin^{2}\theta_1\right],
    \label{eqn:finalayferbound2}\\
    &\frac{1}{2}\log\Bigl(1+\frac{P}{N_1}\Bigr)
    +\min_{\omega_2\in(\frac{\pi}{2}-\theta_2,\frac{\pi}{2}]}
     h_2(\omega_2;\theta_2)\nonumber\\
    =&\frac{1}{2}\log\left[\left(1+
      \frac{P(N_1+N_2)+N_1N_2\cos^{2}\theta_2
            +\sqrt{\bigl[P(N_1+N_2)+N_1N_2\cos^{2}\theta_2\bigr]^{2}
                   -4P^{2}N_1N_2\sin^{2}\theta_2}}
           {2N_1N_2\sin^{2}\theta_2}\right)\sin^{2}\theta_2\right],
    \label{eqn:finalayferbound4}
\end{align}
\end{subequations}
where the minimizers, whose feasibility proof is omitted here, are given by
\begin{align*}
    \cos\omega_1^{\star}
    =&\frac{2P\sqrt{N_1N_2}\,\sin^{2}\theta_1}
           {P(N_1+N_2)+N_1N_2\cos^{2}\theta_1
            +\sqrt{\bigl[P(N_1+N_2)+N_1N_2\cos^{2}\theta_1\bigr]^{2}
                   -4P^{2}N_1N_2\sin^{2}\theta_1}},\\
    \cos\omega_2^{\star}
    =&\frac{2P\sqrt{N_1N_2}\,\sin^{2}\theta_2}
           {P(N_1+N_2)+N_1N_2\cos^{2}\theta_2
            +\sqrt{\bigl[P(N_1+N_2)+N_1N_2\cos^{2}\theta_2\bigr]^{2}
                   -4P^{2}N_1N_2\sin^{2}\theta_2}}.
\end{align*}
Since \eqref{eqn:finalayferbound2} and \eqref{eqn:finalayferbound4} are identical to \eqref{eqn:finalpsibound2} and \eqref{eqn:finalpsibound4} when relating $(\psi_1,\psi_2)$ with $(\theta_1,\theta_2)$, \eqref{eqn:psibound2} and \eqref{eqn:psibound4} coincide with \eqref{eqn:ayfergauss2} and \eqref{eqn:ayfergauss4}, which completes the proof.
\end{proof}

\section{Evaluation for Discrete Channels}\label{sec:discrete}

In this section, we numerically compare the new upper bound in Theorem~\ref{ob:relayuvw} with the cut-set bound, the discrete upper bound by \cite{wu2019new} as stated in Theorem~\ref{ob:ayferdiscrete}, and the DF--CF time-sharing inner bound in Theorem~\ref{thm:df-cf-timesharing}. 
We consider multiple channel instances and set $C_1=C_2=C$ throughout.

\begin{figure}[t]
    \centering
    \begin{subfigure}[t]{0.325\linewidth}
        \centering
        \includegraphics[width=\linewidth]{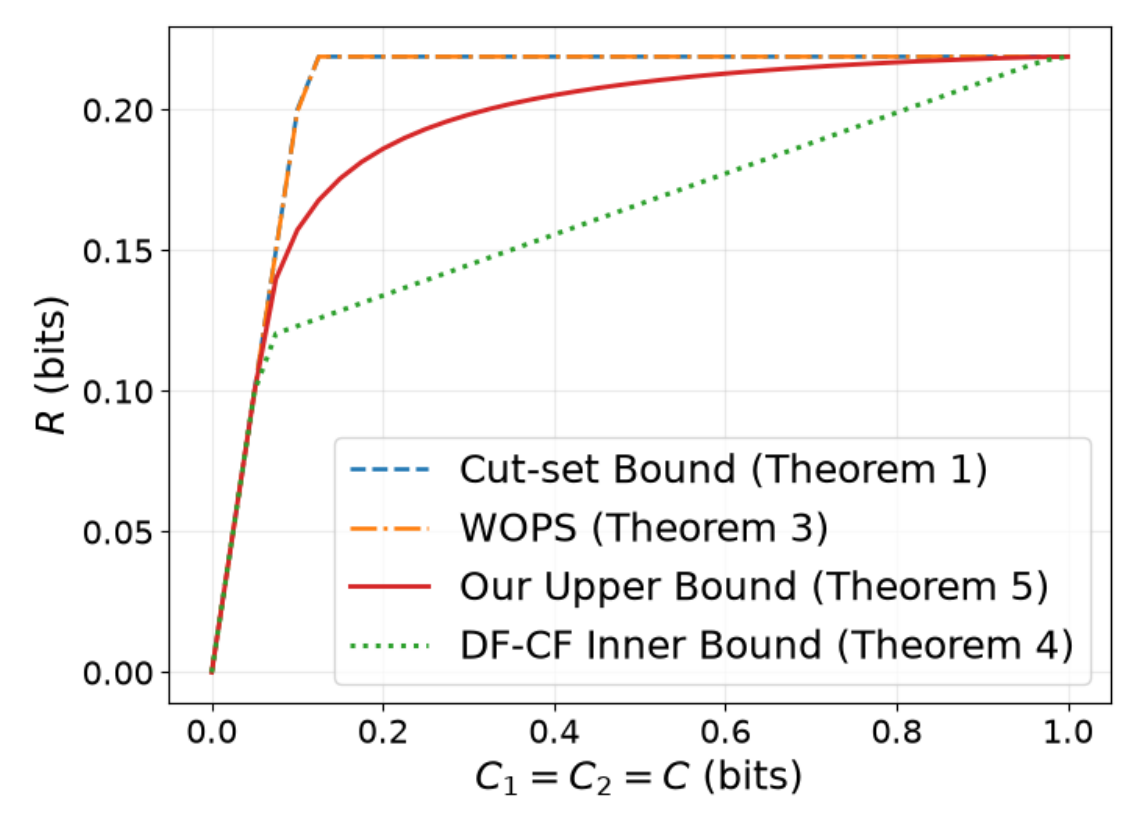}
        \caption{$T_{Y_1|X}=T_{Y_2|X}=\mathrm{BSC}(0.3)$.}
        \label{fig:Ours_symmetricBSC}
    \end{subfigure}
    \hfill
    \begin{subfigure}[t]{0.325\linewidth}
        \centering
        \includegraphics[width=\linewidth]{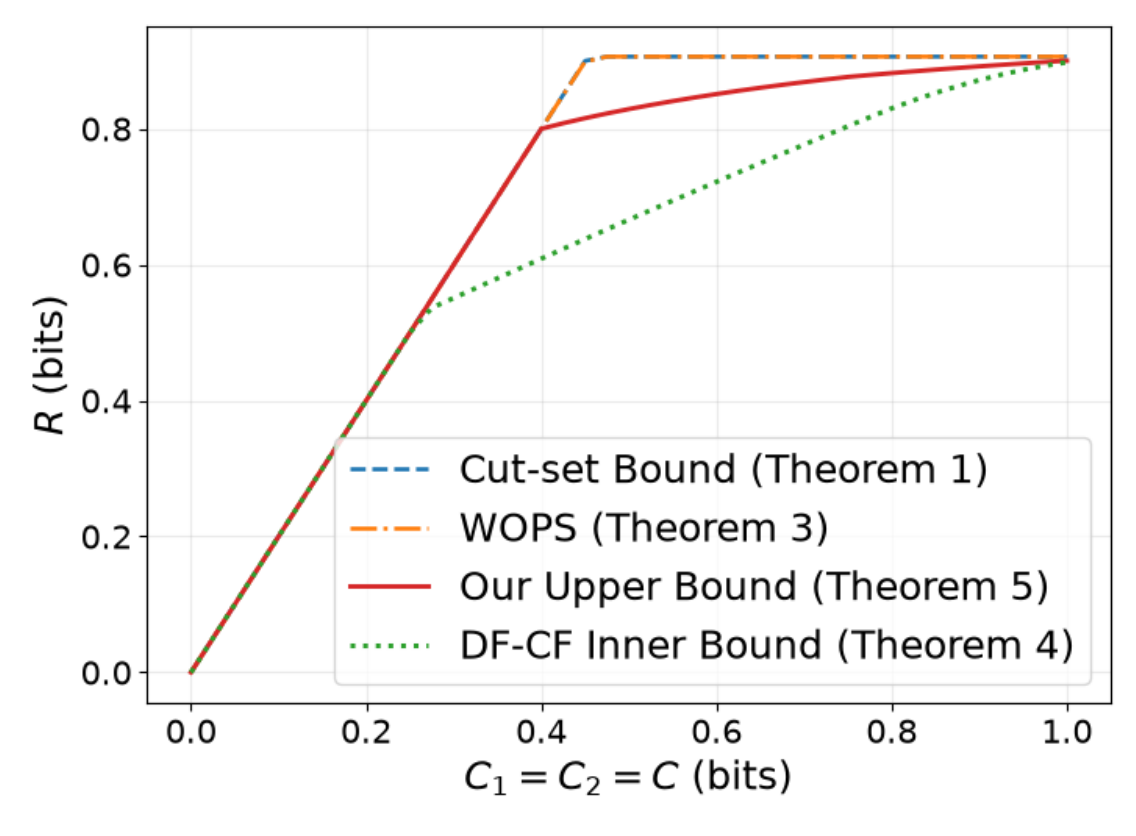}
        \caption{$T_{Y_1|X}=\mathrm{BEC}(0.2)$ and
        $T_{Y_2|X}=\mathrm{BSC}(0.1)$.}
        \label{fig:Ours_BEC_BSC}
    \end{subfigure}
    \hfill
    \begin{subfigure}[t]{0.325\linewidth}
        \centering
        \includegraphics[width=\linewidth]{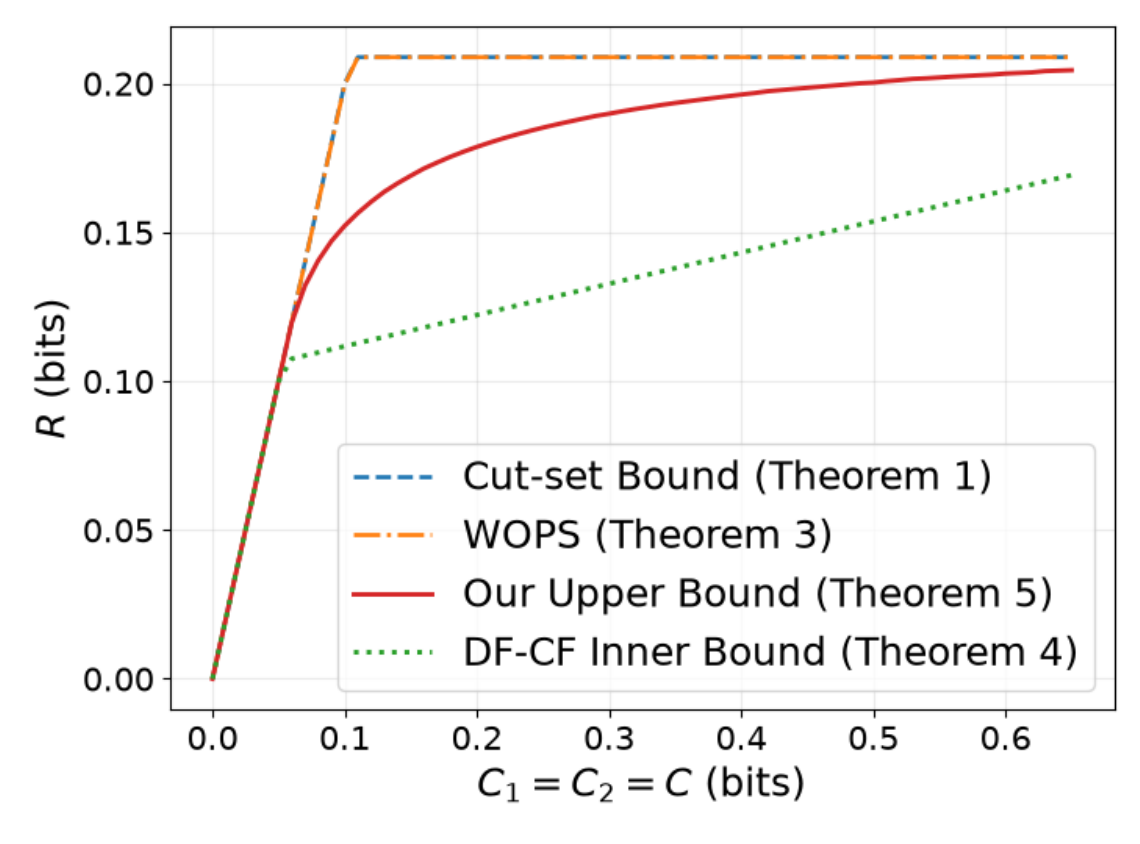}
        \caption{$T_{Y_1|X}=\mathrm{BSC}(0.3)$ and $T_{Y_2|X}$ is obtained by
        cascading $\mathrm{BSC}(0.3)$ with $\mathrm{BEC}(0.1)$.}
        \label{fig:Ours_degrade}
    \end{subfigure}
    \caption{Comparison of the bounds for primitive diamond channels under
    different marginal channels.}
    \label{fig:Ours_three_channels}
\end{figure}

Specifically, we consider the following three instances: 
\begin{itemize}
    \item The symmetric channel $T_{Y_1|X}=T_{Y_2|X}=\mathrm{BSC}(0.3)$. 
    The corresponding comparison is shown in Figure~\ref{fig:Ours_symmetricBSC}. 

    \item The asymmetric channel with $T_{Y_1|X}=\mathrm{BEC}(0.2), T_{Y_2|X}=\mathrm{BSC}(0.1)$, and the corresponding comparison is shown in Figure~\ref{fig:Ours_BEC_BSC}. 

    \item The stochastically degraded marginal channels where $T_{Y_1|X}=\mathrm{BSC}(0.3)$, and $T_{Y_2|X}$ is obtained by cascading a $\mathrm{BSC}(0.3)$ with a $\mathrm{BEC}(0.1)$. The corresponding bounds are shown in Figure~\ref{fig:Ours_degrade}.
\end{itemize}


\section{Concluding Remarks}\label{conclusion}
We derive a single-letter upper bound for the discrete memoryless primitive diamond relay channel. We prove that the Gaussian specialization of the proposed single-letter upper bound coincides with the strengthened Gaussian upper bound of~\cite{wu2019new}. We also numerically compare the proposed upper bound with the cut-set bound, the upper bound of~\cite{wu2019new}, and an optimized DF--CF time-sharing inner bound. For selected discrete memoryless primitive diamond relay channels, the proposed bound is strictly tighter than both existing upper bounds, with a clearly visible numerical gap.

The gap between the known inner and outer bounds motivates the search for a tighter outer bound. The auxiliary-receiver approach has yielded tighter outer bounds for several multiuser channels, including broadcast, interference, and relay channels~\cite{gohari2021outer,gzn24,geng25,gln26}. Although a direct introduction of auxiliary receivers did not improve the outer bound in our investigation, more refined applications of this approach to the primitive diamond relay channel remain a promising direction for future work.

\section*{AI Statement}
The authors used the AI systems GPT-5.6 Sol, GPT-6 Astra, and Claude Fable 5.1 to assist with literature searches, numerical simulations, exploration of the benchmark inner bound, development of the proofs of Lemmas~\ref{lm:gaussianopt} and~\ref{lm:gaussianequivalence}, and manuscript drafting. The authors take full responsibility for the accuracy and integrity of the manuscript.

\appendix
\section*{The Proof of Lemma~\ref{lm:gaussianopt}}\label{app:gaussianopt}
    We show that jointly Gaussian random variables suffice to evaluate the upper bound in Theorem~\ref{ob:relayuvw} for the Gaussian channel. The argument has two parts: we first express the bound as an optimization of weighted information functionals, and then study Gaussian optimality for each such functional.

    For an admissible distribution $p(u,v,x)$, define
    \begin{align*}
        F_1= &I(X;Y_2)+C_1-I(U;Y_1|X),\\
        F_2= &I(X;Y_1)+I(X;Y_2|U)-I(X;Y_1|U),\\
        F_3= & I(X;Y_1)+C_2-I(V;Y_2|X),\\
        F_4= & I(X;Y_2)+I(X;Y_1|V)-I(X;Y_2|V),\\
        F_5= &C_1+C_2-I(U;Y_1|X)-I(V;Y_2|X).
    \end{align*}

    Let $\mathcal P$ denote the class of joint distributions induced by the Gaussian channel for which
$$\mathbb E[X^2]\leq P,\qquad U\rightarrow (X,Y_1)\rightarrow Y_2,\qquad V\rightarrow (X,Y_2)\rightarrow Y_1.$$
Define the set $$\mathcal{A}:= \{a\in \mathcal{R}^5:a_j\leq F_j(p)\quad\forall j, \text{  for some }p\in \mathcal{P} \}.$$

\begin{lemma}\label{lm:convexity-ob}
    The set $\mathcal{A}$ is convex.
\end{lemma}
    \begin{proof}
        Suppose we have $\{a_j^A\}, \{a_j^B\}\in \mathcal{A}$, and $P^A,P^B\in \mathcal{P}$, such that $a_j^A\leq F_j(P^A),a_j^B\leq F_j(P^B),\forall j\in[1:5]$. 
        
        Let $Q\in \{A,B\}$ be Bernoulli($\theta$). We construct a new distribution $r((u,q),(v,q),x)$ such that $\forall (u,v,x)\in\mathcal{U}\times\mathcal{V}\times \mathcal{X}$,
        \begin{align*}
            r((u,A),(v,A),x) =& \theta p^A(u,v,x), \\
            r((u,B),(v,B),x) =& (1-\theta)p^B(u,v,x).
        \end{align*}
        $r\in\mathcal{P}$ since $(U,Q)\rightarrow (X,Y_1)\rightarrow Y_2$, $(V,Q)\rightarrow (X,Y_2)\rightarrow Y_1$, and the power constraint holds as follows
        $$\mathbb E_r[X^2] =\theta \mathbb E_{p^A}[X^2]+(1-\theta) \mathbb E_{p^B}[X^2]\leq P.$$

        In the following, we demonstrate $F_j(r)\geq \theta a_j^A+(1-\theta)a_j^B$ for $j=1,2,5$. The argument for $j=3,4$ follows analogously. Thus, $\mathcal{A}$ is convex.
        \begin{align*}
            F_1(r) =& I_r(X;Y_2)+C_1-I_r(U,Q;Y_1|X)\\
            =&I_r(X;Y_2|Q)+C_1-I_r(U;Y_1|X,Q)+I_r(Q;Y_2)-I_r(Q;Y_1|X)\\
            =&I_r(X;Y_2|Q)+C_1-I_r(U;Y_1|X,Q)+I_r(Q;Y_2)\\
            \geq &I_r(X;Y_2|Q)+C_1-I_r(U;Y_1|X,Q)\\
            =&\theta (I_{p^A}(X;Y_2)+C_1-I_{p^A}(U;Y_1|X))+(1-\theta) (I_{p^B}(X;Y_2)+C_1-I_{p^B}(U;Y_1|X))\\
            =&\theta F_1(p^A)+(1-\theta )F_1(p^B)\\
            \geq &\theta a_1^A+(1-\theta)a_1^B.\\
            F_2(r)=&I_r(X;Y_1)+I_r(X;Y_2|U,Q)-I_r(X;Y_1|U,Q)\\
            =&I_r(X;Y_1|Q)+I_r(X;Y_2|U,Q)-I_r(X;Y_1|U,Q)+I_r(Q;Y_1)\\
            \geq &I_r(X;Y_1|Q)+I_r(X;Y_2|U,Q)-I_r(X;Y_1|U,Q)\\
            =&\theta (I_{p^A}(X;Y_1)+I_{p^A}(X;Y_2|U)-I_{p^A}(X;Y_1|U))+(1-\theta) (I_{p^B}(X;Y_1)+I_{p^B}(X;Y_2|U)-I_{p^B}(X;Y_1|U))\\
            =& \theta F_2(p^A)+(1-\theta) F_2(p^B)\\
            \geq & \theta a_2^A+(1-\theta)a_2^B.\\
            F_5(r)=& C_1+C_2-I_r(U,Q;Y_1|X)-I_r(V,Q;Y_2|X)\\
            =&C_1+C_2-I_r(U;Y_1|X,Q)-I_r(V;Y_2|X,Q)-I_r(Q;Y_1|X)-I_r(Q;Y_2|X)\\
            =&C_1+C_2-I_r(U;Y_1|X,Q)-I_r(V;Y_2|X,Q)\\
            =&\theta (C_1+C_2-I_{p^A}(U;Y_1|X)-I_{p^A}(V;Y_2|X))+(1-\theta) (C_1+C_2-I_{p^B}(U;Y_1|X)-I_{p^B}(V;Y_2|X))\\
            =&\theta F_5(p^A)+(1-\theta) F_5 (p^B)\\
            \geq & \theta a_5^A+(1-\theta) a_5^B.
        \end{align*}
        
    \end{proof}
    
    We rewrite the rate upper bound as
    \begin{align*}
        R\leq &\max_{p\in\mathcal{P}} \min_{1\leq i\leq 5} F_i(p)
        = \max_{p\in\mathcal{P}}\min_{\substack{\lambda_j\geq 0,\\\sum_{j=1}^5 \lambda_j=1}} \sum_{j=1}^5 \lambda_j F_j(p)
        =\min_{\substack{\lambda_j\geq 0,\\\sum_{j=1}^5 \lambda_j=1}} \max_{p\in\mathcal{P}} \sum_{j=1}^5 \lambda_j F_j(p).
    \end{align*}
    The last equality follows from the convexity of $\mathcal{A}$ and the applicable min-max exchange result in \cite{ggny14}[Corollary 2].

    \begin{remark}
        In a similar way, we can prove the set $\mathcal{A}^G:=\{a\in \mathcal{R}^5:a_j\leq F_j(p)\quad\forall j, \text{  for some }p\in \mathcal{P}^G \}$ is convex, where
        $\mathcal P^G$ denotes the subset to  $\mathcal{P}$ with additional requirements such that $P_X$, $P_{X,Y_1|U}$, $P_{X,Y_2|V}$ are Gaussian. Moreover, we conclude the following min-max exchange:
$$\max_{p\in \mathcal{P}^G}\min_{\substack{\lambda_j\geq 0,\\\sum_{j=1}^5 \lambda_j=1}} \sum_{j=1}^5 \lambda_j F_j(p)
        =\min_{\substack{\lambda_j\geq 0,\\\sum_{j=1}^5 \lambda_j=1}} \max_{p\in\mathcal{P}^G} \sum_{j=1}^5 \lambda_j F_j(p).$$
    \end{remark}

    Fix some $\{\lambda_j\geq 0:\sum_{j=1}^5 \lambda_j=1\}$, define
    \begin{align*}
        \mathcal{L}^T(U,V,X) :=&(\lambda_2+\lambda_3)I(X;Y_1)+(\lambda_1+\lambda_4)I(X;Y_2)\\
        &+\lambda_2 I(X;Y_2|U)-\lambda_2 I(X;Y_1|U)-(\lambda_1+\lambda_5)I(U;Y_1|X)\\
        &+\lambda_4 I(X;Y_1|V)-\lambda_4 I(X;Y_2|V)-(\lambda_3+\lambda_5)I(V;Y_2|X),
    \end{align*}
    such that $\mathcal{L}(U,V,X) = \sum_{j=1}^5 \lambda_j F_j-(\lambda_1+\lambda_5)C_1-(\lambda_3+\lambda_5)C_2$.

    Further define
    $$\mathcal V^T = \max_{\substack{\mathbb{E}[X^2]\leq P,\\U\rightarrow(X,Y_1)\rightarrow Y_2,\\V\rightarrow(X,Y_2)\rightarrow Y_1}} \mathcal{L}^T(U,V,X).$$
    It remains to prove that a jointly Gaussian maximizer achieves $\mathcal V^T$.

    We suppose that $(U_1,V_1,X_1)$ and $(U_2,V_2,X_2)$ are two independent copies of the maximizer $(U,V,X)^*$ to $\mathcal L^T(U,V,X)$. We can assume $X^*$ to be zero-mean. Define
    \begin{align*}
        X_{+}=&\frac{X_1+X_2}{\sqrt 2},\quad X_{-} =\frac{X_1-X_2}{\sqrt 2},\\
        Y_{1,+}=&\frac{Y_{1,1}+Y_{1,2}}{\sqrt 2},\quad Y_{1,-} =\frac{Y_{1,1}-Y_{1,2}}{\sqrt 2},\\
        Y_{2,+}=&\frac{Y_{2,1}+Y_{2,2}}{\sqrt 2},\quad Y_{2,-} =\frac{Y_{2,1}-Y_{2,2}}{\sqrt 2}.
    \end{align*}
    The channel is invariant after the above bijective transformation, where the additive noises
    \begin{align*}
        Z_{1,+}=&\frac{Z_{1,1}+Z_{1,2}}{\sqrt 2},\quad Z_{1,-} =\frac{Z_{1,1}-Z_{1,2}}{\sqrt 2},\\
        Z_{2,+}=&\frac{Z_{2,1}+Z_{2,2}}{\sqrt 2},\quad Z_{2,-} =\frac{Z_{2,1}-Z_{2,2}}{\sqrt 2},
    \end{align*}
    are independent and distributed as $\mathcal{N}(0,N_1),\mathcal{N}(0,N_2)$.

    Observe
    \begin{align*}
        2\mathcal{V}^T \overset{(a)}{=}& \mathcal{L}^{T}(U_1,V_1,X_1)+\mathcal{L}^{T}(U_2,V_2,X_2)\\
        \overset{(b)}{=}&\mathcal{L}^{T\otimes T}((U_1,U_2),(V_1,V_2),(X_1,X_2))\\
        =&(\lambda_2+\lambda_3)I(X_1,X_2;Y_{1,1},Y_{1,2})+(\lambda_1+\lambda_4)I(X_1,X_2;Y_{2,1},Y_{2,2})\\
        &+\lambda_2 I(X_1,X_2;Y_{2,1},Y_{2,2}|U_1,U_2)-\lambda_2 I(X_1,X_2;Y_{1,1},Y_{1,2}|U_1,U_2)-(\lambda_1+\lambda_5)I(U_1,U_2;Y_{1,1},Y_{1,2}|X_1,X_2)\\
        &+\lambda_4 I(X_1,X_2;Y_{1,1},Y_{1,2}|V_1,V_2)-\lambda_4 I(X_1,X_2;Y_{2,1},Y_{2,2}|V_1,V_2)-(\lambda_3+\lambda_5)I(V_1,V_2;Y_{2,1},Y_{2,2}|X_1,X_2)\\
        \overset{(c)}{=}&(\lambda_2+\lambda_3)I(X_+,X_-;Y_{1,+},Y_{1,-})+(\lambda_1+\lambda_4)I(X_+,X_-;Y_{2,+},Y_{2,-})\\
        &+\lambda_2 I(X_+,X_-;Y_{2,+},Y_{2,-}|U_1,U_2)-\lambda_2 I(X_+,X_-;Y_{1,+},Y_{1,-}|U_1,U_2)-(\lambda_1+\lambda_5)I(U_1,U_2;Y_{1,+},Y_{1,-}|X_+,X_-)\\
        &+\lambda_4 I(X_+,X_-;Y_{1,+},Y_{1,-}|V_1,V_2)-\lambda_4 I(X_+,X_-;Y_{2,+},Y_{2,-}|V_1,V_2)-(\lambda_3+\lambda_5)I(V_1,V_2;Y_{2,+},Y_{2,-}|X_+,X_-)\\
        \overset{(d)}{=}& (\lambda_2+\lambda_3)I(X_+;Y_{1,+})+(\lambda_1+\lambda_4)I(X_+;Y_{2,+})\\
        &+\lambda_2 (I(X_+;Y_{2,+}|U_1,U_2,Y_{1,-})-I(X_+;Y_{1,+}|U_1,U_2,Y_{1,-}))-(\lambda_1+\lambda_5)I(U_1,U_2,Y_{1,-};Y_{1,+}|X_+)\\
        &+\lambda_4 (I(X_+;Y_{1,+}|V_1,V_2,Y_{2,-})-I(X_+;Y_{2,+}|V_1,V_2,Y_{2,-}))-(\lambda_3+\lambda_5)I(V_1,V_2,Y_{2,-};Y_{2,+}|X_+)\\
        &+(\lambda_2+\lambda_3)I(X_-;Y_{1,-})+(\lambda_1+\lambda_4)I(X_-;Y_{2,-})\\
        &+\lambda_2(I(X_-;Y_{2,-}|U_1,U_2,Y_{2,+})-I(X_-;Y_{1,-}|U_1,U_2,Y_{2,+}))-(\lambda_1+\lambda_5)I(U_1,U_2,Y_{2,+};Y_{1,-}|X_-)\\
        &+\lambda_4(I(X_-;Y_{1,-}|V_1,V_2,Y_{1,+})-I(X_-;Y_{2,-}|V_1,V_2,Y_{1,+}))-(\lambda_3+\lambda_5)I(V_1,V_2,Y_{1,+};Y_{2,-}|X_-))\\
        &-(\lambda_2+\lambda_3)I(Y_{1,+};Y_{1,-})-(\lambda_1+\lambda_4)I(Y_{2,+};Y_{2,-})\\
        &-(\lambda_1+\lambda_2+\lambda_5)(I(Y_{1,+};X_-|X_+,U_1,U_2,Y_{1,-})+I(Y_{1,-};X_+|X_-,U_1,U_2,Y_{2,+}))\\
        &-(\lambda_3+\lambda_4+\lambda_5)(I(Y_{2,+};X_-|X_+,V_1,V_2,Y_{2,-})+I(Y_{2,-};X_+|X_-,V_1,V_2,Y_{1,+}))\\
        \overset{(e)}{\leq }& (\lambda_2+\lambda_3)I(X_+;Y_{1,+})+(\lambda_1+\lambda_4)I(X_+;Y_{2,+})\\
        &+\lambda_2 (I(X_+;Y_{2,+}|U_1,U_2,Y_{1,-})-I(X_+;Y_{1,+}|U_1,U_2,Y_{1,-}))-(\lambda_1+\lambda_5)I(U_1,U_2,Y_{1,-};Y_{1,+}|X_+)\\
        &+\lambda_4 (I(X_+;Y_{1,+}|V_1,V_2,Y_{2,-})-I(X_+;Y_{2,+}|V_1,V_2,Y_{2,-}))-(\lambda_3+\lambda_5)I(V_1,V_2,Y_{2,-};Y_{2,+}|X_+)\\
        &+(\lambda_2+\lambda_3)I(X_-;Y_{1,-})+(\lambda_1+\lambda_4)I(X_-;Y_{2,-})\\
        &+\lambda_2(I(X_-;Y_{2,-}|U_1,U_2,Y_{2,+})-I(X_-;Y_{1,-}|U_1,U_2,Y_{2,+}))-(\lambda_1+\lambda_5)I(U_1,U_2,Y_{2,+};Y_{1,-}|X_-)\\
        &+\lambda_4(I(X_-;Y_{1,-}|V_1,V_2,Y_{1,+})-I(X_-;Y_{2,-}|V_1,V_2,Y_{1,+}))-(\lambda_3+\lambda_5)I(V_1,V_2,Y_{1,+};Y_{2,-}|X_-))\\
        \overset{(f)}{=}&\mathcal{L}^{T}(U_+,V_+,X_+)+\mathcal{L}^{T}(U_-,V_-,X_-)\\
        \overset{(g)}{\leq} &2\mathcal{V}^T,
    \end{align*}
    where $(a)$ holds because $(U,V,X)^*$ achieves $V^T$; $(b)$ holds because two random variable tuples $(U_1,V_1,X_1)$ and $(U_2,V_2,X_2)$ are independent copies; $(c)$ follows since the mutual information terms are invariant under bijective transformation, whose proof can be found in \cite{gn14}; $(d)$ follows by a standard derivation; $(f)$ holds as we do the following substitution:
    \begin{align*}
        &U_+=(U_1,U_2,Y_{1,-}),\quad U_-=(U_1,U_2,Y_{2,+}),\\
        &V_+=(V_1,V_2,Y_{2,-}),\quad V_{-}=(V_1,V_2,Y_{1,+}).
    \end{align*}
    Step $(g)$ follows by the definition of $\mathcal{V}^T$, where the input power constraints are satisfied:
    $$\mathbb{E}[X_+^2]=\mathbb{E}[X_-^2]=\mathbb{E}[X^2]\leq P,$$
    and the Markov chain conditions are satisfied, as well:
    \begin{align*}
        &U_+\rightarrow (X_+,Y_{1,+})\rightarrow Y_{2,+},\quad V_+\rightarrow (X_+,Y_{2,+})\rightarrow Y_{1,+},\\
        &U_-\rightarrow (X_-,Y_{1,-})\rightarrow Y_{2,-},\quad V_-\rightarrow (X_-,Y_{2,-})\rightarrow Y_{1,-}.
    \end{align*}
    Since the first and last expressions in the above sequence of (in)equalities are equal, equality must hold in step~$(e)$. If $\lambda_2+\lambda_3>0$, the nonnegative term $(\lambda_2+\lambda_3)I(Y_{1,+};Y_{1,-})$ must vanish. Thus $Y_{1,+}$ and $Y_{1,-}$ are independent. By \cite[Proposition~2]{gn14}, this implies that $X_+$ and $X_-$ are independent. Since $X_1$ and $X_2$ are independent copies
and $X_\pm=\frac{X_1\pm X_2}{\sqrt{2}}$, the
Skitovich--Darmois theorem implies that $X$ is Gaussian
\cite{Darmois1953,Skitovich1953,GhuryeOlkin1962}.

If $\lambda_1+\lambda_4>0$, the same argument applies to
$I(Y_{2,+};Y_{2,-})$ and again shows that $X$ is Gaussian.
The only remaining case is
$\lambda_1+\lambda_4=\lambda_2+\lambda_3=0$, and $\lambda_5=1$.
In this case,
$$\mathcal L_\lambda(U,V,X)
=-I(U;Y_1\mid X)-I(V;Y_2\mid X)\leq 0,$$
and equality is attained by taking $U$ and $V$ to be constant
and $X$ to be Gaussian.
    
    We define the effective terms related to the variable $U$ as
    $$\mathcal{J}_U^T(U,X) := \lambda_2 I(X;Y_2|U)-\lambda_2 I(X;Y_1|U)-(\lambda_1+\lambda_5) I(U;Y_1|X).$$
    We aim to prove $P_{X,Y_1|U}$ being Gaussian attains $\max_{U\rightarrow (X,Y_1)\rightarrow Y_2}\mathcal{J}_U^T(U,X)$. Note that, when $\lambda_2 = 0$, the maximal can be attained by taking $U$ to be constant and $X$ to be Gaussian. Hence, we focus on when $\lambda_2>0$.

    For $\epsilon>0$, define a perturbed channel $T_\epsilon$, such that
    $$Y_1^\epsilon =\begin{bmatrix}
        Y_1\\\epsilon Z_1+W
    \end{bmatrix},\quad Y_2^\epsilon =\begin{bmatrix}
        Y_2\\\epsilon Z_1+W
    \end{bmatrix},$$
    where $Z_1,Z_2$ denote the channel's additive Gaussian noises, and $W\sim \mathcal{N}_1(0,1)$ is independent of all the other random variables. Define
    \begin{align*}
        \mathcal{J}_U^{\epsilon,T_\epsilon}(U,X) :=&\lambda_2 I(X;Y_2^\epsilon|U)-\lambda_2 I(X;Y_1^\epsilon|U)-(\lambda_1+\lambda_5) I(U;Y_1^\epsilon|X)-\lambda_2 \epsilon h(Y_2^\epsilon|U),\\
        \mathcal{V}_U:=&\max_{U\rightarrow (X,Y_1^\epsilon)\rightarrow Y_2^\epsilon}\mathcal{J}_U^{\epsilon,T_\epsilon} (U,X).
    \end{align*}
     We suppose $(U_1,X_1)$ and $(U_2,X_2)$ are two independent copies of the maximizer to $\mathcal{J}_U^{\epsilon,T_\epsilon} (U,X)$. Define 
     \begin{align*}
        &X_{+}=\frac{X_1+X_2}{\sqrt 2}, && X_{-} =\frac{X_1-X_2}{\sqrt 2},\\
        &Z_{1,+}=\frac{Z_{1,1}+Z_{1,2}}{\sqrt 2}, && Z_{1,-}=\frac{Z_{1,1}-Z_{1,2}}{\sqrt 2},\\
        &Z_{2,+}=\frac{Z_{2,1}+Z_{2,2}}{\sqrt 2},&& Z_{2,-}=\frac{Z_{2,1}-Z_{2,2}}{\sqrt 2},\\
        &W_{+}=\frac{W_1+W_2}{\sqrt 2}, && W_{-} =\frac{W_1-W_2}{\sqrt 2},\\
        &Y_{1,+}^\epsilon=\frac{Y_{1,1}^\epsilon+Y_{1,2}^\epsilon}{\sqrt 2},&& Y_{1,-}^\epsilon =\frac{Y_{1,1}^\epsilon-Y_{1,2}^\epsilon}{\sqrt 2},\\
        &Y_{2,+}^\epsilon=\frac{Y_{2,1}^\epsilon+Y_{2,2}^\epsilon}{\sqrt 2}, && Y_{2,-}^\epsilon =\frac{Y_{2,1}^\epsilon-Y_{2,2}^\epsilon}{\sqrt 2}.
    \end{align*}
    Note that the channel $T_\epsilon$ is again invariant.
     
     Applying a similar argument as the above,
     \begin{align*}
        2 \mathcal{V}_U=& \mathcal{J}_U^{\epsilon,T_\epsilon}(U_1,X_1)+\mathcal{J}_U^{\epsilon,T_\epsilon}(U_2,X_2)\\
        =&\mathcal{J}_U^{\epsilon,T_\epsilon\otimes T_\epsilon}(U_1,U_2,X_1,X_2)\\
        =&\lambda_2 I(X_1,X_2;Y_{2,1}^\epsilon,Y_{2,2}^\epsilon|U_1,U_2)-\lambda_2 I(X_1,X_2;Y_{1,1}^\epsilon,Y_{1,2}^\epsilon|U_1,U_2)\\
        &-(\lambda_1+\lambda_5) I(U_1,U_2;Y_{1,1}^\epsilon,Y_{1,2}^\epsilon|X_1,X_2)-\lambda_2 \epsilon h(Y_{2,1}^\epsilon,Y_{2,2}^\epsilon|U_1,U_2)\\
         =&\lambda_2 I(X_+,X_-;Y_{2,+}^\epsilon,Y_{2,-}^\epsilon|U_1,U_2)-\lambda_2 I(X_+,X_-;Y_{1,+}^\epsilon,Y_{1,-}^\epsilon|U_1,U_2)\\
         &-(\lambda_1+\lambda_5) I(U_1,U_2;Y_{1,+}^\epsilon,Y_{1,-}^\epsilon|X_+,X_-)-\lambda_2 \epsilon h(Y_{2,+}^\epsilon,Y_{2,-}^\epsilon|U_1,U_2)\\
         =&\lambda_2 I(X_+;Y_{2,+}^\epsilon|U_1,U_2,Y_{1,-}^\epsilon)-\lambda_2 I(X_+;Y_{1,+}^\epsilon|U_1,U_2,Y_{1,-}^\epsilon)\\
         &-(\lambda_1+\lambda_5) I(U_1,U_2,Y_{1,-}^\epsilon;Y_{1,+}^\epsilon|X_+)-\lambda_2 \epsilon h(Y_{2,+}^\epsilon|U_1,U_2,Y_{1,-}^\epsilon)\\
         &+\lambda_2 I(X_-;Y_{2,-}^\epsilon|U_1,U_2,Y_{2,+}^\epsilon)-\lambda_2 I(X_-;Y_{1,-}^\epsilon|U_1,U_2,Y_{2,+}^\epsilon)\\
         &-(\lambda_1+\lambda_5) I(U_1,U_2,Y_{2,+}^\epsilon;Y_{1,-}^\epsilon|X_-)-\lambda_2 \epsilon h(Y_{2,-}^\epsilon|U_1,U_2,Y_{2,+}^\epsilon)\\
         &-(\lambda_1+\lambda_2+\lambda_5)(I(X_-;Y_{1,+}^\epsilon|X_+,U_1,U_2,Y_{1,-}^\epsilon)+I(X_+;Y_{1,-}^\epsilon|X_-,U_1,U_2,Y_{2,+}^\epsilon))\\
         &-\lambda_2\epsilon I(Y_{1,-}^\epsilon;Y_{2,+}^\epsilon|U_1,U_2)\\
         \leq &\lambda_2 I(X_+;Y_{2,+}^\epsilon|U_1,U_2,Y_{1,-}^\epsilon)-\lambda_2 I(X_+;Y_{1,+}^\epsilon|U_1,U_2,Y_{1,-}^\epsilon)\\
         &-(\lambda_1+\lambda_5) I(U_1,U_2,Y_{1,-}^\epsilon;Y_{1,+}^\epsilon|X_+)-\lambda_2 \epsilon h(Y_{2,+}^\epsilon|U_1,U_2,Y_{1,-}^\epsilon)\\
         &+\lambda_2 I(X_-;Y_{2,-}^\epsilon|U_1,U_2,Y_{2,+}^\epsilon)-\lambda_2 I(X_-;Y_{1,-}^\epsilon|U_1,U_2,Y_{2,+}^\epsilon)\\
         &-(\lambda_1+\lambda_5) I(U_1,U_2,Y_{2,+}^\epsilon;Y_{1,-}^\epsilon|X_-)-\lambda_2 \epsilon h(Y_{2,-}^\epsilon|U_1,U_2,Y_{2,+}^\epsilon)\\
         =& \mathcal{J}_U^{\epsilon,T_\epsilon}(U_+,X_+)+\mathcal{J}_U^{\epsilon,T_\epsilon}(U_-,X_-)\\
         \leq & 2 \mathcal{V}_U,
     \end{align*}
     where we do the following substitution:
     $$U_+=(U_1,U_2,Y_{1,-}^\epsilon),\quad U_-=(U_1,U_2,Y_{2,+}^\epsilon).$$
     By squeezing, we have 
    \begin{align*}
        0=&I(Y_{1,-}^\epsilon;Y_{2,+}^\epsilon|U_1,U_2)\\
        =&I(\begin{pmatrix}
            1 & 1\\ 0 &\epsilon
        \end{pmatrix}\begin{pmatrix}
            X_-\\Z_{1,-}
        \end{pmatrix}+\begin{pmatrix}
            0 \\W_-
        \end{pmatrix};\begin{pmatrix}
            1 & 0\\ 0 &\epsilon
        \end{pmatrix} \begin{pmatrix}
            X_+\\Z_{1,+}
        \end{pmatrix}+\begin{pmatrix}
            Z_{2,+}\\W_+
        \end{pmatrix} |U_1,U_2)\\
        =&I(\begin{pmatrix}
            X_-\\Z_{1,-}
        \end{pmatrix};\begin{pmatrix}
            X_+\\Z_{1,+}
        \end{pmatrix}|U_1,U_2).
    \end{align*}
    Again, Skitovich-Darmois theorem\cite{Darmois1953,Skitovich1953,GhuryeOlkin1962} therefore implies that $P_{X,Z_1|U}$ is Gaussian. By a limiting argument, we further conclude that the optimizer $P_{X,Y_1|U}$ is Gaussian for the original channel. A symmetric argument works for proving the optimizer $P_{X,Y_2|V}$ is Gaussian.

\bibliographystyle{IEEEtran}
\bibliography{ref}

@inproceedings{aoe25,
  title={Neural compress-and-forward for the primitive diamond relay channel},
  author={Ayg{\"u}n, Ozan and {\"O}zyilkan, Ezgi and Erkip, Elza},
  booktitle={2025 59th Asilomar Conference on Signals, Systems, and Computers},
  pages={754--758},
  year={2025},
  organization={IEEE}
}

@INPROCEEDINGS{gdy25,
  author={Guo, Yuanxin and Draper, Stark and Yu, Wei},
  booktitle={IEEE International Symposium on Information Theory (ISIT)}, 
  title={Modulo Quantization Coding for Gaussian Primitive Relay Channel With Perfectly Correlated Noises}, 
  year={2025},
  volume={},
  number={},
  pages={1-6},
  doi={10.1109/ISIT63088.2025.11195654}}

@unknown{gdy26,
author = {Guo, Yuanxin and Draper, Stark and Yu, Wei},
year = {2026},
month = {06},
pages = {},
title = {Modulo Quantization Coding for Primitive Relay and Diamond Channels with Correlated Noises},
doi = {10.48550/arXiv.2606.22313}
}

@INPROCEEDINGS{geng25,
  author={Geng, Yanlin},
  booktitle={2025 IEEE International Symposium on Information Theory (ISIT)}, 
  title={A Simple Auxiliary Receiver Outer Bound for Broadcast Channels}, 
  year={2025},
  volume={},
  number={},
  pages={1-6},
  doi={10.1109/ISIT63088.2025.11195345}}

@INPROCEEDINGS{gzn24,
  author={Gohari, Amin and Nair, Chandra and Zhao, Jinpei},
  booktitle={2024 IEEE International Symposium on Information Theory (ISIT)}, 
  title={On the capacity region of some classes of interference channels}, 
  year={2024},
  volume={},
  number={},
  pages={3136-3141},
  doi={10.1109/ISIT57864.2024.10619605}}

@INPROCEEDINGS{gln26,
  author={Gohari, Amin and Liu, Yi and Nair, Chandra},
  booktitle={IEEE International Symposium on Information Theory (ISIT)}, 
  title={A Two Auxiliary Receiver Outer Bound to the Capacity Region of a Two-Receiver Discrete Memoryless Broadcast Channel}, 
  year={2026},
  volume={},
  number={},
  pages={1-6},
  doi={10.1109/ISIT62367.2026.11653844}}

@article{SanderovichShamaiSteinbergKramer2008,
  author  = {Amichai Sanderovich and Shlomo Shamai (Shitz) and
             Yossef Steinberg and Gerhard Kramer},
  title   = {Communication via Decentralized Processing},
  journal = {IEEE Transactions on Information Theory},
  volume  = {54},
  number  = {7},
  pages   = {3008--3023},
  month   = jul,
  year    = {2008},
  doi     = {10.1109/TIT.2008.924659}
}

@article{SanderovichShamaiSteinberg2009,
  author  = {Amichai Sanderovich and Shlomo Shamai (Shitz) and
             Yossef Steinberg},
  title   = {Distributed {MIMO} Receiver---Achievable Rates
             and Upper Bounds},
  journal = {IEEE Transactions on Information Theory},
  volume  = {55},
  number  = {10},
  pages   = {4419--4438},
  month   = oct,
  year    = {2009}
}

@article{SanderovichSomekhPoorShamai2009,
  author  = {Amichai Sanderovich and Oren Somekh and
             H. Vincent Poor and Shlomo Shamai (Shitz)},
  title   = {Uplink Macro Diversity of Limited Backhaul
             Cellular Network},
  journal = {IEEE Transactions on Information Theory},
  volume  = {55},
  number  = {8},
  pages   = {3457--3478},
  month   = aug,
  year    = {2009}
}

@INPROCEEDINGS{wbo17,
  author={Wu, Xiugang and Barnes, Leighton Pate and Özgür, Ayfer},
  booktitle={2017 IEEE International Symposium on Information Theory (ISIT)}, 
  title={The geometry of the relay channel}, 
  year={2017},
  volume={},
  number={},
  pages={2233-2237},
  doi={10.1109/ISIT.2017.8006926}}

@article{bwo23,
author = {Bai, Yikun and Wu, Xiugang and Özgür, Ayfer},
year = {2023},
month = {04},
pages = {2059-2073},
title = {Information Constrained Optimal Transport: From Talagrand, to {M}arton, to {C}over},
volume = {69},
journal = {IEEE Transactions on Information Theory},
doi = {10.1109/TIT.2023.3237073}
}

@ARTICLE{wbo19,
  author={Wu, Xiugang and Barnes, Leighton Pate and Özgür, Ayfer},
  journal={IEEE Transactions on Information Theory}, 
  title={“{T}he Capacity of the Relay Channel”: Solution to {C}over's Problem in the {G}aussian Case}, 
  year={2019},
  volume={65},
  number={1},
  pages={255-275},
  doi={10.1109/TIT.2018.2876892}}

@article{CoverElGamal1979,
  author  = {Thomas M. Cover and Abbas A. El Gamal},
  title   = {Capacity Theorems for the Relay Channel},
  journal = {IEEE Transactions on Information Theory},
  volume  = {25},
  number  = {5},
  pages   = {572--584},
  month   = sep,
  year    = {1979},
  doi     = {10.1109/TIT.1979.1055985}
}

@article{LimKimElGamalChung2011NNC,
  author  = {Sung Hoon Lim and Young-Han Kim and
             Abbas El Gamal and Sae-Young Chung},
  title   = {Noisy Network Coding},
  journal = {IEEE Transactions on Information Theory},
  volume  = {57},
  number  = {5},
  pages   = {3132--3152},
  month   = may,
  year    = {2011},
  doi     = {10.1109/TIT.2011.2119930}
}

@ARTICLE{GK26,
  author={Gohari, Amin and Kramer, Gerhard},
  journal={IEEE Transactions on Information Theory}, 
  title={Dependence Balance and Capacity Bounds for Multiterminal Communication and Wiretap Channels}, 
  year={2026},
  volume={},
  number={},
  pages={1-1},
  doi={10.1109/TIT.2026.3707564}}

@book{bunt1934bijdrage,
  title={Bijdrage tot de theorie der convexe puntverzamelingen},
  author={Bunt, Lucas Nicolaas Hendrik}, 
  year={1934},
  publisher={Verlag nicht ermittelbar}
}

@ARTICLE{ggny14,
  author={Geng, Yanlin and Gohari, Amin and Nair, Chandra and Yu, Yuanming},
  journal={IEEE Transactions on Information Theory}, 
  title={On {M}arton's Inner Bound and Its Optimality for Classes of Product Broadcast Channels}, 
  year={2014},
  volume={60},
  number={1},
  pages={22-41},
  doi={10.1109/TIT.2013.2285925}}

@article{GhuryeOlkin1962,
  author  = {Ghurye, S. G. and Olkin, Ingram},
  title   = {A Characterization of the Multivariate Normal Distribution},
  journal = {The Annals of Mathematical Statistics},
  volume  = {33},
  number  = {2},
  pages   = {533--541},
  month   = jun,
  year    = {1962},
  doi     = {10.1214/aoms/1177704579}
}

@article{Darmois1953,
  author  = {Darmois, Georges},
  title   = {Analyse g{\'e}n{\'e}rale des liaisons stochastiques},
  journal = {Revue de l'Institut International de Statistique},
  volume  = {21},
  pages   = {2--8},
  year    = {1953}
}

@article{Skitovich1953,
  author  = {Skitovich, V. P.},
  title   = {On a Property of the Normal Distribution},
  journal = {Doklady Akademii Nauk SSSR},
  volume  = {89},
  pages   = {217--219},
  year    = {1953},
  note    = {In Russian}
}

@ARTICLE{gn14,
  author={Geng, Yanlin and Nair, Chandra},
  journal={IEEE Transactions on Information Theory}, 
  title={The Capacity Region of the Two-Receiver {G}aussian Vector Broadcast Channel With Private and Common Messages}, 
  year={2014},
  volume={60},
  number={4},
  pages={2087-2104},
  doi={10.1109/TIT.2014.2304457}}

@INPROCEEDINGS{egn21,
  author={El Gamal, Abbas and Gohari, Amin and Nair, Chandra},
  booktitle={IEEE International Symposium on Information Theory (ISIT)}, 
  title={Strengthened Cutset Upper Bound on the Capacity of the Relay Channel and Applications}, 
  year={2021},
  volume={},
  number={},
  pages={1344-1349},
  doi={10.1109/ISIT45174.2021.9517975}}

@article{gohari2021outer,
  title={Outer bounds for multiuser settings: The auxiliary receiver approach},
  author={Gohari, Amin and Nair, Chandra},
  journal={IEEE Transactions on Information Theory},
  volume={68},
  number={2},
  pages={701--736},
  year={2021},
  publisher={IEEE}
}

@article{bidokhti2016capacity,
  title={Capacity bounds for diamond networks with an orthogonal broadcast channel},
  author={Bidokhti, Shirin Saeedi and Kramer, Gerhard},
  journal={IEEE Transactions on Information Theory},
  volume={62},
  number={12},
  pages={7103--7122},
  year={2016},
  publisher={IEEE}
}

@article{kang2015gaussian,
  title={The {G}aussian multiple access diamond channel},
  author={Kang, Wei and Liu, Nan and Chong, Weiwei},
  journal={IEEE Transactions on Information Theory},
  volume={61},
  number={11},
  pages={6049--6059},
  year={2015},
  publisher={IEEE}
}

@inproceedings{wu2019new,
  title={New upper bounds on the capacity of primitive diamond relay channels},
  author={Wu, Xiugang and Ozgur, Ayfer and Peleg, Michael and Shitz, Shlomo Shamai},
  booktitle={2019 IEEE Information Theory Workshop (ITW)},
  pages={1--5},
  year={2019},
  organization={IEEE}
}

@phdthesis{schein2001distributed,
  title={Distributed coordination in network information theory},
  author={Schein, Brett Eric},
  year={2001},
  school={Massachusetts Institute of Technology}
}

@article{aguerri2019capacity,
  title={On the capacity of cloud radio access networks with oblivious relaying},
  author={Aguerri, Inaki Estella and Zaidi, Abdellatif and Caire, Giuseppe and Shamai, Shlomo},
  journal={IEEE Transactions on Information Theory},
  volume={65},
  number={7},
  pages={4575--4596},
  year={2019},
  publisher={IEEE}
}

@article{liu2026nonasymptotic,
  title={Nonasymptotic Oblivious Relaying and Variable-Length Noisy Lossy Source Coding},
  author={Liu, Yanxiao and Advary, Sepehr Heidari and Li, Cheuk Ting}, 
  journal={IEEE Transactions on Information Theory},
  volume={72},
  number={6},
  pages={4555--4564},
  year={2026},
  publisher={IEEE}
}

@article{park2014fronthaul,
  title={Fronthaul compression for cloud radio access networks: Signal processing advances inspired by network information theory},
  author={Park, Seok-Hwan and Simeone, Osvaldo and Sahin, Onur and Shitz, Shlomo Shamai},
  journal={IEEE Signal Processing Magazine},
  volume={31},
  number={6},
  pages={69--79},
  year={2014},
  publisher={IEEE}
}

@article{wu2015cloud,
  title={Cloud radio access network ({C-RAN}): a primer},
  author={Wu, Jun and Zhang, Zhifeng and Hong, Yu and Wen, Yonggang},
  journal={IEEE Network},
  volume={29},
  number={1},
  pages={35--41},
  year={2015},
  publisher={IEEE}
}

\end{document}